\documentclass[letterpaper,twocolumn,10pt,conference]{ieeeconf}
\usepackage{amsmath}
\usepackage{tikz}
\usepackage{amsfonts,empheq}
\usepackage{amscd}
\usepackage{amssymb}
\usepackage{subfigure}
\usepackage{balance}
\newtheorem{lem}{Lemma}[section]
\newtheorem{theorem}{Theorem}[section]
\IEEEoverridecommandlockouts
\begin{document}
\title{Can Iterative Decoding for Erasure Correlated Sources be Universal?}
\author{Arvind Yedla, Henry D.~Pfister, Krishna R.~Narayanan\thanks{This work was supported in part by the National Science Foundation under Grant No. CCR-0515296 and by the Qatar National Research Foundation under its National Research Priorities Program.}\\
{\normalsize Department of Electrical and Computer Engineering, Texas A\&M University}\\
{\normalsize College Station, TX 77840, U.S.A}\\
{\normalsize \{yarvind,hpfister,krn\}@tamu.edu}}

\maketitle
\thispagestyle{empty}
\pagestyle{empty}

  \begin{abstract}
    In this paper, we consider a few iterative decoding schemes for the
    joint source-channel coding of correlated sources. Specifically, we
    consider the joint source-channel coding of two erasure correlated
    sources with transmission over different erasure channels. Our main
    interest is in determining whether or not various code ensembles can
    achieve the capacity region universally over varying channel
    conditions. We consider two ensembles in the class of
    low-density generator-matrix (LDGM) codes known as Luby-Transform
    (LT) codes and one ensemble of low-density parity-check (LDPC) codes.
    We analyze them using density evolution and show that
    optimized LT codes can achieve the extremal symmetric point of the
    capacity region. We also show that LT codes are not universal under
    iterative decoding for this problem because they cannot
    simultaneously achieve the extremal symmetric point and a corner
    point of the capacity region. The sub-universality of iterative
    decoding is characterized by studying the density evolution for LT
    codes.
  \end{abstract}

  \section{Introduction}
  \label{intro}
  The system model considered in this paper is shown in Figure~\ref{fig:sys_model}. We wish to transmit the
  outputs of two discrete memoryless correlated sources $\left(U_i^{(1)},U_i^{(2)}\right)$, for
  $i=1,2,\cdots,k$ to a central receiver through two independent discrete memoryless channels with capacities
  $C_1$ and $C_2$, respectively. We will assume that each channel can be parameterized by a single parameter $\epsilon_i$ for $i=1,2$
  (e.g., the erasure probability or crossover probability).
  The two sources are not allowed to collaborate and, hence, they use two independent
  encoding functions which map the $k$ input symbols in to $n_1$ and $n_2$ output symbols, respectively. The rates of
  the encoders are given by $R_1 = k/n_1$ and $R_2 = k/n_2$.

  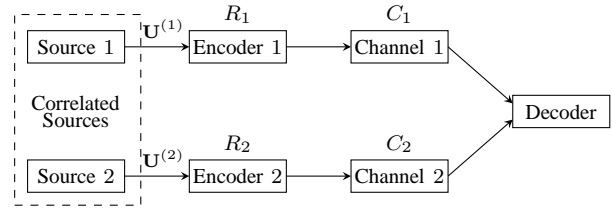
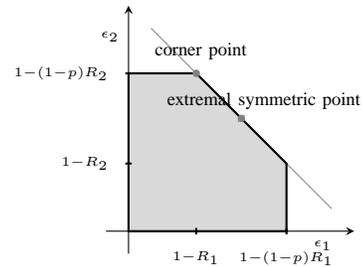
\begin{figure}[!h]
    \centering 
    \subfigure[System Model]{
\begin{tikzpicture}[scale=0.43,>=stealth,xshift=-1cm]
\draw (0,0) rectangle +(3,1);
\draw (1.5,.5) node {\footnotesize Source $2$};
\draw (0,4) rectangle +(3,1);
\draw (1.5,4.5) node {\footnotesize Source $1$};
\draw (5,0) rectangle +(3,1);
\draw (6.5,.5) node {\footnotesize Encoder $2$};
\draw (6.5,1.5) node {\footnotesize $R_2$};
\draw (5,4) rectangle +(3,1);
\draw (6.5,4.5) node {\footnotesize Encoder $1$};
\draw (6.5,5.5) node {\footnotesize $R_1$};
\draw (10,0) rectangle +(3,1);
\draw (11.5,.5) node {\footnotesize Channel $2$};
\draw (11.5,1.5) node {\footnotesize $C_2$};
\draw (10,4) rectangle +(3,1);
\draw (11.5,4.5) node {\footnotesize Channel $1$};
\draw (11.5,5.5) node {\footnotesize $C_1$};
\draw (15,2) rectangle +(3,1);
\draw (16.5,2.5) node {\footnotesize Decoder};

\draw[->] (3,0.5) -- (5,0.5) node[pos=0.6,above] {$\scriptstyle{\mathbf{U}^{(2)}}$};
\draw[->] (3,4.5) -- (5,4.5) node[pos=0.6,above] {$\scriptstyle{\mathbf{U}^{(1)}}$};
\draw[->] (8,0.5) -- (10,0.5);
\draw[->] (8,4.5) -- (10,4.5);
\draw[->] (13,0.5) -- (15,2.25);
\draw[->] (13,4.5) -- (15,2.75);
\draw[dashed] (-0.4,-0.4) rectangle +(3.9,5.9);
\draw (1.5,2.75) node {\footnotesize Correlated};
\draw (1.5,2.25) node {\footnotesize Sources};
\end{tikzpicture}

      \label{fig:sys_model}
    }
    \subfigure[The achievable region for the $2$ user case]{
\begin{tikzpicture}[>=stealth,scale=0.6]
\draw[gray, very thin] (0.5,4.5) -- (4.5,0.5); 
\draw[->] (-0.5,0) -- (5,0) node[very near end,sloped,below] {$\scriptscriptstyle\epsilon_1$};
\draw[->] (0,-0.5) -- (0,5) node[very near end,left] {$\scriptscriptstyle\epsilon_2$};
\shade[top color=gray!30!white, bottom color=gray!30!white] (0,0) -- (0,3.5) -- (1.5,3.5) -- (3.5,1.5) -- (3.5,0) -- cycle;
\draw[thick] (0,0) -- (0,3.5) -- (1.5,3.5) -- (3.5,1.5) -- (3.5,0) -- cycle;
\draw[gray, very thin] (2.425,2.425) -- (2.575,2.575);
\draw[thick] (3.5cm,-2pt) -- (3.5cm,2pt) node[below=5pt] {$\scriptscriptstyle 1-(1-p)R_1$};
\draw[thick] (1.5cm,-2pt) -- (1.5cm,2pt) node[below=5pt] {$\scriptscriptstyle{1-R_1}$};
\draw[thick] (-2pt,3.5cm) -- (2pt,3.5cm) node[left=5pt] {$\scriptscriptstyle 1-(1-p)R_2$};
\draw[thick] (-2pt,1.5cm) -- (2pt,1.5cm) node[left=5pt] {$\scriptscriptstyle{1-R_2}$};
 \draw (1.6,3.6) node[above] {\scriptsize corner point};
 \filldraw[gray] (1.5,3.5) circle (2pt);
 \draw (3,2.5) node[above] {\scriptsize extremal symmetric point};
 \filldraw[gray] (2.5,2.5) circle (2pt);
\end{tikzpicture}

      \label{fig:roc_center}
    }
    \caption{System Model}
  \end{figure}

  In such a problem,
  it is clear that one has to take advantage of the correlation between the sources to reduce the required bandwidth to transmit
  the information to the central receiver. Thus, this joint source-channel coding problem can be seen to be an instance of Slepian-Wolf
  coding \cite{Slepian-it73} in the presence of a noisy channel.
  If $\epsilon_1$ and $\epsilon_2$ are known to transmitter 1 and 2 respectively, then
the sources can be reliably decoded at the receiver iff
   \begin{equation}
    \label{eq:sw}
    \begin{split}
      \frac{C_1 (\epsilon_1)}{R_1} &\geq \textnormal{H}\left(U^{(1)}|U^{(2)}\right) \\
      \frac{C_2 (\epsilon_2)}{R_2} &\geq \textnormal{H}\left(U^{(2)}|U^{(1)}\right) \\
      \frac{C_1 (\epsilon_1)}{R_1}+\frac{C_2 (\epsilon_2)}{R_2}&\geq \textnormal{H}\left(U^{(1)},U^{(2)}\right).
    \end{split}
  \end{equation}
  In this case, one can separate the problem into Slepian-Wolf coding \cite{Slepian-it73} of the two sources
and channel coding for the two channels.
In recent years, there have been graph based coding schemes which, under iterative decoding, can
  obtain near optimal performance for this problem \cite{garciafrias2001jsc,zhong2005lcc,liveris2002jsc, liveris2002cbs}.

  However, in several practical situations, it is
  unrealistic for the transmitters to have a priori knowledge of $\epsilon_1$ and $\epsilon_2$. Therefore, we consider the case where the transmitters each use
  a single code of rate $R$ (though it is possible to extend this to different rates $R_1$ and $R_2$).
  We then wish to find a universal source-channel coding scheme such that reliable transmission
  is possible over a range of channel parameters $(\epsilon_1,\epsilon_2)$. Ideally, we would like to have one code of
  rate $R_1 = R_2 = R$ that allows error free communication of the sources for any set of channel parameters $(\epsilon_1,\epsilon_2)$ for
  which $\epsilon_1,\epsilon_2$ satisfy the conditions in (\ref{eq:sw}).
  For a given pair of encoding functions of rate $R$ and a joint decoding algorithm, the achievable channel parameter region (ACPR) is defined as the set of all channel parameters 
  $(\epsilon_1,\epsilon_2)$ for which the
  encoder/decoder combination achieves an arbitrarily low probability of error as $k \rightarrow \infty$. For some channels, this region is equal to the entire region in (\ref{eq:sw}) and, in this case, we call it
  the capacity region. Note that the ACPR and the capacity region are defined as
  the set of all channel parameters for which successful recovery of the sources is possible for a fixed encoding
  rate pair $(R,R)$ (or, more generally $(R_1,R_2)$) instead of the set of rates $(R_1,R_2)$ for a fixed pair of channel parameters $(\epsilon_1,\epsilon_2)$.

  It can be seen that the capacity region is, in fact, given by all pairs of $(\epsilon_1,\epsilon_2)$ such that (\ref{eq:sw}) is satisfied.
  For binary-input memoryless symmetric channels, this region is achieved when both users encode with independent random linear codes and use maximum-likelihood (or typical set) decoding at the receiver. This means that random codes with ML decoding are universal for symmetric channels.
  That is, for a given $(R_1,R_2)$, a single encoder/decoder pair suffices to communicate the sources over all pairs of symmetric channels
  for which $(\epsilon_1,\epsilon_2)$ satisfy the conditions in
  (\ref{eq:sw}). Thus, one can obtain optimal performance even without knowledge of $(\epsilon_1,\epsilon_2)$ at the transmitter.
  We refer to such encoder/decoder pairs as being {\em universal}.

  While random codes with ML decoding are universally good, this scheme is clearly impractical due to its large complexity.
  Our primary interest in this paper is to investigate whether there exist graph based codes and iterative decoding algorithms that
  are also universal and to find good encoder/decoder pairs that result in large ACPRs.
  Several code ensembles, including Luby Transform (LT) codes and LDPC codes, have been
  shown to achieve capacity with iterative decoding on a single user erasure channel \cite{Luby-focs02,RU-2008}. However, the universality of these
  ensembles for more complicated scenarios has not been studied well in the literature. Hence, the question of whether one can
  design a single graph based code and a decoding algorithm capable of universal performance is a question that has not been answered
  in the literature. One of the main results in this paper is that iterative decoding of LT codes cannot be universal thus showing that
  ensembles that are good for single user channels do not necessarily perform well for the joint source-channel coding problem. \par

  Before we discuss the main results in this paper in Section~\ref{sec:summary-results}, we first introduce a specific instance of the
  problem described above which is simple and yet captures the difficulty of designing a universal joint source-channel coding scheme.


    \section{System Model for Erasure Correlation}
  \label{sec:sys-model}
  Consider the case where the source correlation and channels both have
  an erasure structure. Let $Z_i$, for $i=1,2,\cdots,k$, denoted by the
  column vector $\mathbf{Z}_k$, be a sequence of i.i.d. Bernoulli-$p$
  random variables.  The correlation between $\mathbf{U}^{(1)}$ and
  $\mathbf{U}^{(2)}$ is defined by
  \begin{equation*}
    \left(U^{(1)}_i,U^{(2)}_i\right) = \left\{ \begin{array}{l}
        \textnormal{i.i.d. Bernoulli $\frac{1}{2}$ r.v.s}, \textnormal{if } Z_i=0\\\\
        \textnormal{same Bernoulli $\frac{1}{2}$ r.v. $U_i$ }, \textnormal{if } Z_i=1
      \end{array} \right.
  \end{equation*}
  We consider transmission over erasure channels with erasure rates $\epsilon_1$
  and $\epsilon_2$. The Slepian-Wolf conditions are satisfied if
  \begin{equation*}
    \begin{split}
      (1-\epsilon_1) &\geq (1-p)R_1 \\
      (1-\epsilon_2) &\geq (1-p)R_2\\
      \frac{(1-\epsilon_1)}{R_1} +  \frac{(1-\epsilon_2)}{R_2} &\geq 2-p,
    \end{split}
  \end{equation*}
  and the achievable channel parameter region is shown in Fig. \ref{fig:roc_center}. \par
  The source sequences $\mathbf{U}^{(1)}$ and $\mathbf{U}^{(2)}$ are
  encoded using a pair of independent binary linear codes
  $\mathcal{C}_1[n,k]$ and $\mathcal{C}_2[n,k]$ chosen from the same
  code ensemble.  We consider the encoding and decoding of LDPC codes
  and LDGM codes separately.

  \subsection{LDGM codes}
  \label{sec:ldgm-codes}
  The source sequences are encoded using different LDGM codes chosen
  from the same ensemble, defined in terms of generator matrices
  $G^{(1)}$ and $G^{(2)}$. The encoded sequences denoted by
  $\mathbf{X}^{(1)}$ and $\mathbf{X}^{(2)}$ are given by
  \begin{equation*}
    \mathbf{X}^{(i)} =
    \begin{bmatrix}
      \mathbf{U}^{(i)} \\
      G^{(i)^T}\mathbf{U}^{(i)}
    \end{bmatrix}.
  \end{equation*}
  The source bits $\mathbf{U}^{(1)}$ and $\mathbf{U}^{(2)}$ are
  punctured and then transmitted through binary erasure channels (BECs)
  with erasure rates $\epsilon_1$ and $\epsilon_2$ respectively. The
  governing equations at the decoder are given by
  \begin{equation*}
    \begin{bmatrix}
      G^{(i)^T} & I_n
    \end{bmatrix}
    \mathbf{X}^{(i)} = \mathbf{0},\text{ for } i=1,2,
  \end{equation*}
  where $I_n$ is an $n\times n$ identity matrix. For simplicity
  of notation, we define $H^{(i)} = \begin{bmatrix} G^{(i)^T} & I_n
  \end{bmatrix}$ for $i=1,2$, for the case of LDGM codes.  Given a
  matrix $A$, and a suitable index set $\mathcal{I}$, let
  $A_{\mathcal{I}}$ ($A_{\mathcal{I}'}$) denote the sub-matrix of $A$,
  restricted to the columns (rows) indexed by $\mathcal{I}$. Let
  $\mathcal{P}$ denote the set of indices corresponding to the non-zero
  locations of $\mathbf{Z}_k$, and $Z$ be the diagonal matrix, whose
  diagonal is given by $\left[\mathbf{Z}_k, \mathbf{0}\right]$, where
  $\mathbf{0}$ denotes a vector of all zeros of appropriate length. The
  governing equation $H\mathbf{X} = \mathbf{0}$ at the joint decoder can
  therefore be written in terms of the stacked parity check matrix
  \begin{equation}
    \label{eq:2}
    H =
    \begin{bmatrix}
      H^{(1)} & 0 \\
      0 & H^{(2)} \\
      Z_{\mathcal{P}'} & Z_{\mathcal{P}'}
    \end{bmatrix},
  \end{equation}
  where $\mathbf{X} = \left[\mathbf{X}^{(1)},\,\mathbf{X}^{(2)}\right]$
  and $\left[\,\cdot\,,\,\cdot\,\right]$ denotes concatenation.

  \subsection{LDPC codes}
  \label{sec:ldpc-codes-1}
  The source sequences are encoded using LDPC codes, defined in terms
  of parity-check matrices $H^{(1)}$ and $H^{(2)}$. The encoded
  sequences, denoted  $\mathbf{X}^{(1)}$ and $\mathbf{X}^{(2)}$, are
  encoded using a punctured systematic encoder and transmitted through
  binary erasure channels (BECs) with erasure rates $\epsilon_1$ and
  $\epsilon_2$ respectively. The governing equations at the decoder
  are given by
  \begin{equation*}
    H^{(i)}\mathbf{X}^{(i)} = \mathbf{0},\text{ for } i=1,2.
  \end{equation*}
  For joint decoding, the governing equations (including the source
  correlation constraints), written in terms of the stacked parity
  check matrix defined in (\ref{eq:2}), are given by
  \begin{equation*}
    H\mathbf{X} = \mathbf{0},
  \end{equation*}
  where  $\mathbf{X} = \left[\mathbf{X}^{(1)},\,\mathbf{X}^{(2)}\right]$.


  \subsection{Maximum Likelihood Block Decoder}
  \label{sec:maxim-likel-block}
  In the case of an erasure channel, ML decoding of linear codes is
  equivalent to solving systems of linear equations, which can be
  performed using Gaussian elimination. Let $\mathcal{E}_1$,
  $\mathcal{E}_2$ $\left(\text{and }\bar{\mathcal{E}}_1,
    \bar{\mathcal{E}}_2\right)$ denote the index sets of erasures (and
  non-erasures) corresponding to the received vectors, and let
  $\mathcal{E} = \left[\mathcal{E}_1,\,k+n+\mathcal{E}_2\right]$,
  $\bar{\mathcal{E}} =
  \left[\bar{\mathcal{E}}_1,\,k+n+\bar{\mathcal{E}}_2\right]$.  Denote the
  received sequences by $\mathbf{Y}^{(1)}$ and $\mathbf{Y}^{(2)}$ with
  $\mathbf{Y} = \left[\mathbf{Y}^{(1)},\,\mathbf{Y}^{(2)}\right]$. For
  the binary case, the defining equation $H\mathbf{Y}=0$ simplifies to
  $H_{\mathcal{E}}\mathbf{Y}_{\mathcal{E}}=H_{\bar{\mathcal{E}}}\mathbf{Y}_{\bar{\mathcal{E}}}$,
  in this case. Block ML decoding will be successful iff
  $H_{\mathcal{E}}$ has full rank and the erasures can be recovered by
  inverting $H_{\mathcal{E}}$.

  \subsection{Example}
  \label{sec:example}
  For example, consider the case where $k=4$ and $n=3$ using the LDPC framework.
  Then, we can choose $\mathcal{C}_1$ and $\mathcal{C}_2$ be 
  $\left[7,4\right]$ Hamming codes. If $\mathbf{Z}_4
  = \begin{bmatrix} 1 & 0 & 0 & 1\end{bmatrix}^T$, then the stacked
  parity-check matrix is given by
  \begin{equation*}
    H =
    \left[\begin{array}{ccccccc|ccccccc}
      1 & 1 & 0 & 1 & 1 & 0 & 0 & 0 & 0 & 0 & 0 & 0 & 0 & 0 \\
      1 & 0 & 1 & 1 & 0 & 1 & 0 & 0 & 0 & 0 & 0 & 0 & 0 & 0 \\
      0 & 1 & 1 & 1 & 0 & 0 & 1 & 0 & 0 & 0 & 0 & 0 & 0 & 0 \\
      \hline 0 & 0 & 0 & 0 & 0 & 0 & 0 & 1 & 1 & 0 & 1 & 1 & 0 & 0 \\
      0 & 0 & 0 & 0 & 0 & 0 & 0 & 1 & 0 & 1 & 1 & 0 & 1 & 0 \\
      0 & 0 & 0 & 0 & 0 & 0 & 0 & 0 & 1 & 1 & 1 & 0 & 0 & 1 \\
      \hline 1 & 0 & 0 & 0 & 0 & 0 & 0 & 1 & 0 & 0 & 0 & 0 & 0 & 0 \\
      0 & 0 & 0 & 1 & 0 & 0 & 0 & 0 & 0 & 0 & 1 & 0 & 0 & 0 
    \end{array}\right]
  \end{equation*}
  The Tanner graph corresponding to the
  stacked parity check matrix for LDGM codes is shown in
  Fig. \ref{fig:tanner} and iterative decoding is performed on
  this Tanner graph.


  \section{Outline of the paper and summary of results}
  \label{sec:summary-results}
  We now summarize the main results of this paper.
  \begin{itemize}
    \item In Section~\ref{sec:analysis}, we consider the design and analysis of LDGM codes. We derive the density evolution equations for LT codes
    in Section~\ref{sec:dens-evol-equat}. In Section~\ref{sec:linear-programming}, we consider the design of LDGM codes for the extremal symmetric point and the corner point of the capacity region using linear programming.

  \item  In Section~\ref{sec:symmetric-sum-rate}, we first show
  analytically that LT codes with iterative decoding can achieve the extremal symmetric point of
    the ACPR. However, they cannot achieve the corner point of the ACPR and, hence,
    LT codes with iterative decoding cannot be universal for the joint source-channel coding problem.

  \item In Section~\ref{sec:performance}, we show from simulations that LT codes and the (4,6) LDPC code using maximum likelihood decoding are nearly universal.
  \end{itemize}
These results essentially show that the problem in obtaining universality with the LT ensemble is essentially with the decoding algorithm
rather than with code ensemble. This motivates us to find other decoding algorithms such as enhancements to
message passing decoding that are nearly universal or to consider other code ensembles than the LT code ensemble
with iterative decoding.

    \section{Design and Analysis of LDGM Codes}
  \label{sec:analysis}

  \subsection{Density Evolution Equations}
  \label{sec:dens-evol-equat}
  Assume that the sequences $U^{(1)}$ and $U^{(2)}$ are encoded using LT
  codes with degree distribution pairs
  $\left(\lambda^{(i)},\rho^{(i)}\right)$, for $i=1,2$.
  Based on standard notation \cite{RU-2008}, for $i=1,2$, we let
  $\lambda^{(i)}(x) = \sum_j \lambda_j^{(i)} x^{j-1}$ be the
  degree distribution (from an edge perspective) corresponding to the
  information variable nodes and $\rho^{(i)}(x) = \sum_j  \rho_j^{(i)} x^{j-1}$
  be the degree distribution (from an edge perspective) of the generator (aka check) nodes
  in the decoding graph.  The coefficient $\lambda_j^{(i)}$ (resp. $\rho_j^{(i)}$)
  gives the fraction of edges that connect to the information variable nodes
  (resp. generator nodes) of degree $j$.
  Likewise, $L^{(i)}(x) = \sum_j L^{(i)}_j x^j$ (resp. $R^{(i)}(x) = \sum_j R^{(i)}_j x^j$) is the
degree distributions from the node perspective and $L^{(i)}_j$ (resp. $R^{(i)}_j$) is the
fraction of information variable (resp. generator) nodes with degree $j$.

  Since the encoded variable nodes are are attached to generator nodes randomly,
the degree of a each information variable is a Poisson random variable whose mean
is given by the average number of edges attached to each variable node.  This mean is given by
$\alpha_i = R'^{(i)}(1) / R_i$, where $R'^{(i)}(1)$ is the average generator (or check) degree.
Therefore, the resulting degree distribution is
$L^{(i)}(x) = \text{e}^{\alpha_i(x-1)}$ for $i=1,2$.

  The Tanner graph \cite{RU-2008} for the code is shown in Fig. \ref{fig:tanner},
  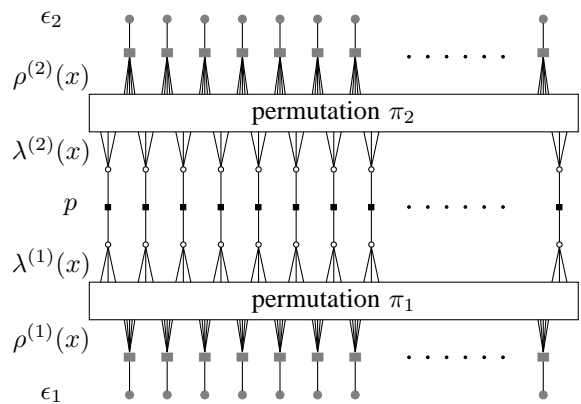
\begin{figure}[!h]
    \centering
\begin{tikzpicture}[scale=0.5]
\draw (-1,2) rectangle +(13,1);
\draw (5.5,2.5) node {permutation $\pi_1$};
\draw (-1,7) rectangle +(13,1);
\draw (5.5,7.5) node {permutation $\pi_2$};
\foreach \x in {0,1,2,3,4,5,6,11}
{
\filldraw[gray] (\x,0)+(2pt,0pt) circle (3pt);
\draw (\x,0)+(2pt,3pt) -- ([xshift=2pt]\x,1); 
\filldraw[gray] (\x,0.9)+(-2pt,0pt) rectangle +(6pt,6pt);
\draw (\x,1)+(2pt,4pt) -- (\x,2);
\draw (\x,1)+(2pt,4pt) -- ([xshift = 4pt]\x,2);
\draw (\x,1)+(2pt,4pt) -- ([xshift = 2pt]\x,2);
\draw (\x,1)+(2pt,4pt) -- ([xshift = -2pt]\x,2);
\draw (\x,1)+(2pt,4pt) -- ([xshift = 6pt]\x,2);
\draw (\x,4)+(0.5,0) circle (2pt);
\draw (\x,4)+(0.5cm,-2pt) -- ([xshift=0.5cm]\x,3);
\draw (\x,4)+(0.5cm,-2pt) -- ([xshift=0.3cm]\x,3);
\draw (\x,4)+(0.5cm,-2pt) -- ([xshift=0.7cm]\x,3);

\draw (\x,4)+(0.5cm,2pt) -- ([xshift=0.5cm]\x,5.925);

\draw (\x,6)+(0.5,0) circle (2pt);
\draw (\x,6)+(0.5cm,2pt) -- ([xshift=0.5cm]\x,7);
\draw (\x,6)+(0.5cm,2pt) -- ([xshift=0.3cm]\x,7);
\draw (\x,6)+(0.5cm,2pt) -- ([xshift=0.7cm]\x,7);
\draw (\x,9)+(2pt,0pt) -- (\x,8);
\draw (\x,9)+(2pt,0pt) -- ([xshift = 4pt]\x,8);
\draw (\x,9)+(2pt,0pt) -- ([xshift = 2pt]\x,8);
\draw (\x,9)+(2pt,0pt) -- ([xshift = -2pt]\x,8);
\draw (\x,9)+(2pt,0pt) -- ([xshift = 6pt]\x,8);
\filldraw[gray] (\x,10)+(2pt,0pt) circle (3pt);
\draw (\x,10)+(2pt,-3pt) -- ([xshift=2pt]\x,9.15); 
\filldraw[gray] (\x,9)+(-2pt,0pt) rectangle +(6pt,6pt);
}
\draw (-0.5,4) circle (2pt);
\draw (-0.5,4)+(0cm,-2pt) -- ([xshift=0.5cm]-1,3);
\draw (-0.5,4)+(0cm,-2pt) -- ([xshift=0.3cm]-1,3);
\draw (-0.5,4)+(0cm,-2pt) -- ([xshift=0.7cm]-1,3);

\draw (-1,4)+(0.5cm,2pt) -- ([xshift=0.5cm]-1,5.925);

\draw (-0.5,6) circle (2pt);
\draw (-0.5,6)+(0cm,2pt) -- ([xshift=0.5cm]-1,7);
\draw (-0.5,6)+(0cm,2pt) -- ([xshift=0.3cm]-1,7);
\draw (-0.5,6)+(0cm,2pt) -- ([xshift=0.7cm]-1,7);
\foreach \x in {0.425,1.425,2.425,...,6.425,11.425} {
\filldraw (\x,4.925) rectangle +(4pt,4pt);
}
\filldraw (-0.575,4.925) rectangle +(4pt,4pt);
\foreach \x in {7.5,8,8.5,...,10} {
\foreach \y in {1,5,9} {
\filldraw (\x,\y) circle (1pt);
}}
\draw (-2,0) node {$\epsilon_1$};
\draw (-2,1.5) node {$\rho^{(1)}(x)$};
\draw (-2,3.5) node {$\lambda^{(1)}(x)$};
\draw (-1.5,5) node {$p$};
\draw (-2,10) node {$\epsilon_2$};
\draw (-2,8.5) node {$\rho^{(2)}(x)$};
\draw (-2,6.5) node {$\lambda^{(2)}(x)$};
\end{tikzpicture}

    \caption{Tanner Graph of an LT Code with erasure correlation between the sources}
    \label{fig:tanner}
  \end{figure}
  from which the density evolution equations \cite{RU-2008} in terms of
  the generator-node to variable-node messages ($x_i$ and $y_i$
  corresponding to codes $1$ and $2$) can be written as follows
  \begin{equation*}
    \begin{split}
      x_{i+1} &= 1 - (1-\epsilon_1)\rho^{(1)}\left(1 -\left((1-p)+pL^{(2)}(y_i)\right)\lambda^{(1)}(x_i)\right)\\
      y_{i+1} &= 1 - (1-\epsilon_2)\rho^{(2)}\left(1 -\left((1-p)+pL^{(1)}(x_i)\right)\lambda^{(2)}(y_i)\right),
    \end{split}
  \end{equation*}
  where $L^{(i)}(x)$, for $i=1,2$, are the degree distributions (from
  the node perspective) corresponding to the information bits. For
  analysis, it is easier to consider the evolution of
  the variable-node to generator-node messages, given by
  \begin{equation*}
    \begin{split}
      x_{i+1} &= \left[(1-p)+pL^{(2)}\left(\varrho^{(2)}(\epsilon_2,y_i)\right)\right]\lambda^{(1)}\left(\varrho^{(1)}(\epsilon_1,x_i)\right)\\
      y_{i+1} &= \left[(1-p)+pL^{(1)}\left(\varrho^{(1)}(\epsilon_1,x_i)\right)\right]\lambda^{(2)}\left(\varrho^{(2)}(\epsilon_2,y_i)\right),
    \end{split}
  \end{equation*}
  where $\varrho^{(i)}(\epsilon,x) = 1 -
  (1-\epsilon)\rho^{(i)}(1-x)$.  Notice that, for LT codes, the
  variable-node degree distribution from the edge perspective is given
  by $\lambda^{(i)}(x) = L^{(i)}(x)$ because $\lambda(x) \triangleq
  L'(x)/L'(1)$ for Poisson $L(x)$.  With this simplification, the
  density evolution equations can be written as
  \begin{equation*}
    \begin{split}
      x_{i+1} &= \left[(1-p)+p\lambda^{(2)}\left(\varrho^{(2)}(\epsilon_2,y_i)\right)\right]\lambda^{(1)}\left(\varrho^{(1)}(\epsilon_1,x_i)\right)\\
      y_{i+1} &= \left[(1-p)+p\lambda^{(1)}\left(\varrho^{(1)}(\epsilon_1,x_i)\right)\right]\lambda^{(2)}\left(\varrho^{(2)}(\epsilon_2,y_i)\right).
    \end{split}
  \end{equation*}

  \subsection{Optimization of degree distributions via Linear Programming}
  \label{sec:linear-programming}
  We use linear programming to design two LT codes.  The first code, called LT code I,
  is designed using the successful decoding constraints for the extremal symmetric point,
  given by the channel condition $(\epsilon,\epsilon) =
  \left(1-\frac{2-p}{2}R,1 - \frac{2-p}{2}R\right)$, as follows.
  \begin{itemize}
  \item Choose the maximum check degree to be $N$.
  \item Compute $\alpha = \frac{1 + G_N(p)}{1-\epsilon}$, with
    $G_N(p)$ as defined in (\ref{eq:gnp_def}).
  \item Maximize $\sum_i\rho_i/i$, subject to, $\forall\;x\in[0,1]$,
    \begin{equation}
      \label{eq:lp}
      \hspace{-3mm}
      \sum_{1\leq i\leq N} \rho_i \cdot \left( 1 - \left((1-p)+p\kappa(\epsilon,x)\right)\kappa(\epsilon,x)\right)^{i-1} < x,
    \end{equation}
where $ \kappa(\epsilon,x) = \text{e}^{\alpha(1-\epsilon)(x-1)}$. 

  \end{itemize}
  The constraints in (\ref{eq:lp}) are obtained from the density
  evolution equations, in terms of the generator-node to variable-node
  messages, described in Section \ref{sec:symmetric-sum-rate} (the
  messages correspond to a modified Tanner graph, where all the
  generator nodes corresponding to the erasures in the received sequence
  have been removed). To achieve a corner point in the Slepian-Wolf
  region, given by the channel condition $(\epsilon_1,\epsilon_2) =
  \left(1-(1-p)R,1 - R\right)$, the constraints in
  (\ref{eq:corner-point-constraints}) were added (obtained from the
  density evolution equations described in \ref{sec:corner-point},
  assuming that the code corresponding to the better channel has
  converged).  This gives, $\forall \; x\in [0,1]$,
  \begin{align}
    \label{eq:corner-point-constraints}
    \sum_{1\leq i\leq N} \rho_i\cdot \left(1 -
    \left((1-p)+p\kappa(\epsilon_1,0)\right)\kappa(\epsilon_2,x)\right)^{i-1} \! < \! x.
  \end{align}

  \subsection{The extremal symmetric point}
  \label{sec:symmetric-sum-rate}
  We first analyze a code  optimized for the case when both
  channels have the same erasure probability
  ($\epsilon_1=\epsilon_2=\epsilon$), to understand the criteria for
  achieving universality. Due to the symmetry of the model for this
  case, we have $\rho^{(1)}(x)=\rho^{(2)}(x)=\rho(x)$ and
  $\lambda^{(1)}(x)=\lambda^{(2)}(x)=\lambda(x) =
  \text{e}^{\alpha(x-1)}$, and the density evolution equations collapse into a
  one-dimensional recursion, given by
  \begin{equation*}
    \label{eq:one-de} 
      x_{i+1} \! = \! \bigl[(1-p)+p\lambda\bigl(1-(1-\epsilon)\rho(1-x_i)\bigr)\bigr]\lambda\bigl(1-(1-\epsilon)\rho(1-x_i)\bigr).
  \end{equation*}
  This recursion can be solved analytically, resulting in the unique non-negative $\rho(x)$ which satisfies 
  \begin{equation*}
    \label{eq:one-deeq} 
      x = \bigl[(1-p)+p\lambda\bigl(1-(1-\epsilon)\rho(1-x)\bigr)\bigr]\lambda\bigl(1-(1-\epsilon)\rho(1-x)\bigr).
  \end{equation*}
  The solution is given by
  \begin{equation*}
    \begin{split}
      \rho(x) &= \frac{-1}{\alpha(1-\epsilon)}\cdot\log\left(\frac{\sqrt{(1-p)^2+4p(1-x)}-(1-p)}{2p}\right)\\
      &= \frac{1}{\alpha(1-\epsilon)}\sum_{i=1}^{\infty}
      \frac{\sum_{k=0}^{i-1}\binom{2i-1}{k}p^k}{i(1+p)^{2i-1}}x^i,
    \end{split}
  \end{equation*}
  and we note that is not a valid degree distribution because it has infinite
  mean. To overcome this, we define a truncated version of the generator degree
  distribution via
  \begin{equation}
    \label{eq:gnp_def}
    \begin{split}
      \rho^N(x) &= \frac{\mu+\sum_{i=1}^{N} \frac{\sum_{k=0}^{i-1}\binom{2i-1}{k}p^k}{i(1+p)^{2i-1}}x^i+x^N}{\mu+G_N(p)+1}\\
      G_N(p) &=
      \sum_{i=1}^N\frac{\sum_{k=0}^{i-1}\binom{2i-1}{k}p^k}{i(1+p)^{2i-1}},
    \end{split}
  \end{equation}
  for some $\mu>0$ and $N\in \mathbb{N}$.  This is a well defined degree
  distribution as all the coefficients are non-negative and
  $\rho^N(1)=1$. The parameter $\mu$ increases the number of degree one
  generator nodes and is introduced in order to overcome the stability problem at
  the beginning of the decoding process \cite{Luby-focs02}.\newline
  \begin{theorem}
    \label{thm-sum-rate}
    Consider transmission over erasure channels with parameters
    $\epsilon_1 = \epsilon_2 = \epsilon$. Let $N\in \mathbb{N}$ and $\mu>0$ and
    \begin{align*}
      \alpha &= \frac{\mu + G_N(p) + 1}{1 -\epsilon},
      \intertext{where}
      G_N(p) &= \sum_{i=1}^N\frac{\sum_{k=0}^{i-1}\binom{2i-1}{k}p^k}{i(1+p)^{2i-1}}.
    \end{align*}
    Then, in the limit of infinite blocklengths, the ensemble
    LDGM$\left(n,\lambda(x),\rho^N(x)\right)$, where
    \begin{equation}
      \label{eq:ensemble}
      \begin{split}
      \lambda(x) &= \text{e}^{\alpha(x-1)},\\
      \rho^N(x) &= \frac{\mu+\sum_{i=1}^{N} \frac{\sum_{k=0}^{i-1}\binom{2i-1}{k}p^k}{i(1+p)^{2i-1}}x^i+x^N}{\mu+G_N(p)+1},
      \end{split}
    \end{equation}
    enables transmission at a rate $R =
    \frac{(1-\epsilon)(1-\text{e}^{-\alpha})}{\mu+1-p/2}$, with a bit
    error probability not exceeding $1/N$.
  \end{theorem}
  \begin{proof}
    See Appendix~\ref{sec:proof-theorem-refthm}.
  \end{proof}
  From Theorem \ref{thm-sum-rate}, we conclude that optimized LT codes,
  given by the ensemble LDGM$\left(n,\lambda(x),\rho^N(x)\right)$
  can achieve the extremal symmetric point of the capacity region.

  \subsection{A Corner Point}
  \label{sec:corner-point}
  Consider the performance of the ensemble
  LDGM$\left(n,\lambda(x),\rho^{(N)}(x)\right)$, with $\lambda(x)$ and
  $\rho^{(N)}(x)$ as defined in (\ref{eq:ensemble}), at a corner point
  of the Slepian-Wolf region. One corner point is given by the channel
  condition $(\epsilon_1,\epsilon_2) =
  \left(1-(1-p)R,1-(1-p/2)R\right)$. The density evolution equations are
  \begin{equation}
    \label{eq:decorner}
    \begin{split}
      x_{i+1} &=       \left[(1-p)+p\bar{\lambda}^N(\epsilon_2,y_i)\right]\bar{\lambda}^N(\epsilon_1,x_i)\\
      y_{i+1} &=       \left[(1-p)+p\bar{\lambda}^N(\epsilon_1,x_i)\right]\bar{\lambda}^N(\epsilon_2,y_i),\\
    \end{split}
  \end{equation}
  where
  $\bar{\lambda}^N(\epsilon,x)=\lambda\left(1-(1-\epsilon)\rho^N(1-x)\right)$. \newline
  \begin{theorem}
    \label{non-universal}
    LT codes cannot simultaneously achieve the extremal symmetric point
    and a corner point of the Slepian-Wolf region, under iterative
    decoding. 
  \end{theorem}
  \begin{proof}
    See Appendix~\ref{sec:proof-theorem-refnon}.
  \end{proof}
  From Theorem \ref{non-universal}, we conclude that LT codes designed
  for the extremal symmetric point are not universal for the two-user
  Slepian-Wolf problem, with erasure correlated sources.


    \section{Performance of Various Code Ensembles}
  \label{sec:performance}
  In this section, we study the performance of three code ensembles under iterative
  and maximum likelihood decoding using simulations. The codes considered are 
  \begin{enumerate}
  \item A linear code with a random generator matrix.
  \item A $(4,6)$ regular LDPC code with punctured systematic bits.
  \item Two LT codes (LT code I and LT code II) optimized for different
    points in the capacity region.
  \end{enumerate}
  LT code I is optimized for the case when both channels have the same
  erasure probability (i.e., the extremal symmetric point of the capacity
  region). LT code II is optimized for the extremal symmetric point,
  including constraints corresponding to channel conditions at one
  corner of the capacity region. Joint iterative decoding is performed
  on the Tanner graph corresponding to the stacked parity check matrix
  $H$. The simplified message passing rules for the BEC are used. They
  are stated here for convenience. At a variable node, the outgoing
  message is an erasure if all incoming messages are
  erasures. Otherwise, all non-erasure messages must have the same
  value, and the outgoing message is equal to the common value. At the
  check node, the outgoing message is an erasure if any of the incoming
  messages is an erasure.
  Otherwise, the outgoing message is the XOR of all the
  incoming messages.  Joint ML decoding was performed on the stacked
  parity check matrix as described in Section \ref{sec:ldgm-codes}.

  The simulations were performed with codes of rate $1/2$ (i.e., two
  encoded bits are generated per source bit), and a blocklength of
  $500$. We chose a source correlation of $p=0.5$, and simulated $300$
  blocks for each point in the capacity region. All the plots are shown
  in the $(\epsilon_1,\epsilon_2)$-plane, for the rate pair $(1/2,1/2)$.

\subsection{Random Codes}
\label{sec:random-codes}
Two different codes of rate $1/2$ are chosen randomly from the
generator-matrix ensemble, where the entries of the generator matrix
are i.i.d. Bernoulli-$1/2$ random variables. Decoding was performed on
the stacked parity-check matrix corresponding to LDGM codes. The ACPR
of random codes under iterative and ML decoding is shown in
Fig. \ref{fig:rand_ml}, respectively. As expected, random codes
achieve the entire capacity region under ML decoding, but perform very
poorly under iterative decoding. The ACPR with iterative decoding
consists of only $3$ non-trivial points with channel parameters very
close to zero.
  \begin{figure}[!t]
    \centering
    \includegraphics[width=0.35\textwidth]{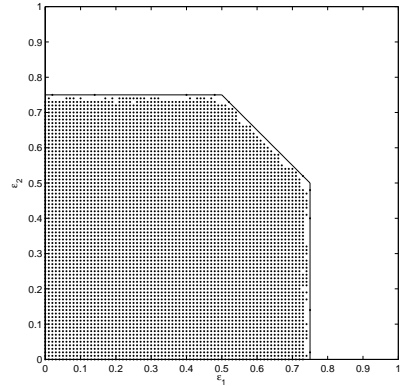}
    \label{fig:rand_ml}
  
  \caption{ACPR for a Random Code under ML Decoding}
\end{figure}

\subsection{LT Codes}
\label{sec:lt-codes}
LT codes have been shown to be universal for the single-user erasure
channel. Here, we study the performance of LT codes for the two-user
erasure channel and consider the case of encoding and decoding at the
extremal symmetric point of the Slepian-Wolf region. An LT code is
optimized for this point (LT Code I), using linear programming (see
Section~\ref{sec:linear-programming}), resulting in the  degree
distribution given by
\begin{equation*}
  \begin{split}
  \label{eq:optimized-corner}
  \rho(x) =\; &0.0001 + 0.0754\cdot x + 0.0295\cdot x^{2} + 0.0620\cdot x^{3} + \\ 
              &0.0857\cdot x^{7} + 0.0718\cdot x^{15} + 0.0970\cdot x^{31} + \\
              &0.0114\cdot x^{63} + 0.5671\cdot x^{127}.
\end{split}
\end{equation*}
The performance of this code under iterative decoding is shown in
Fig. \ref{fig:lt1_it}. Also shown in Fig. \ref{fig:lt1_it} is the
simulated density evolution threshold for LT code I. The density
evolution threshold at the extremal symmetric point is away from
capacity for this code due to limiting the maximum check degree in the
design process. Also, note that the density evolution threshold is far
away from capacity at the corner points of the capacity region. On the
other hand, as seen in Fig. \ref{fig:lt1_ml}, the code performs much
better under ML decoding and is closer to capacity at the corner points
of the capacity region. This reinforces the conclusion that most of the
sub-universality is due to the iterative decoder, rather than the
stability of the code.\par
In order to achieve capacity at the corner points of the capacity
region, LT code II was designed by adding constraints corresponding to
the channel conditions at a corner point of the capacity region (see
Section~\ref{sec:linear-programming}), resulting in the following degree
distribution, 
\begin{equation*}
  \begin{split}
  \label{eq:optimized-dd-corner}
  \rho(x) =\; &0.0001 + 0.0640\cdot x + 0.0251\cdot x^{2} + 0.0526\cdot x^{3} + \\
              &0.0725\cdot x^{7} + 0.0619\cdot x^{15} + 0.0806\cdot x^{31} + \\
              &0.0082\cdot x^{63} + 0.6351\cdot x^{127}.
\end{split}
\end{equation*}
The performance of this code under iterative decoding (and the simulated
density evolution threshold) ML decoding is shown in
Fig. \ref{fig:lt2_it} and Fig. \ref{fig:lt2_ml} respectively. Note that
the density evolution threshold increases only marginally, and the
performance under ML decoding is almost the same.
\begin{figure}[!b]
  \centering
  \subfigure[ACPR for LT Code I under Iterative Decoding]{
    \includegraphics[width=0.35\textwidth]{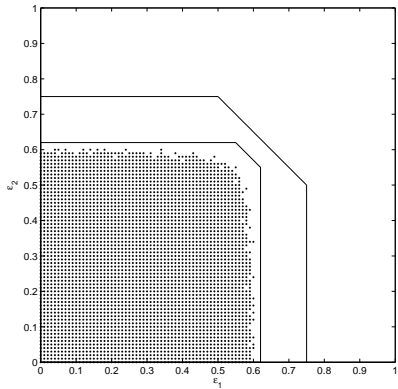}
    \label{fig:lt1_it}
  }
  \subfigure[ACPR for LT Code I under ML Decoding]{
    \includegraphics[width=0.35\textwidth]{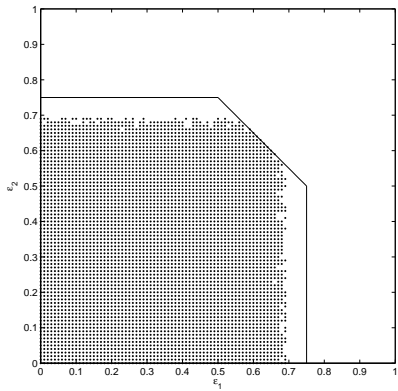}
    \label{fig:lt1_ml}
  }
  \caption{Performance of LT Code I}
\end{figure}

\begin{figure}[htb]
  \centering
  \subfigure[ACPR for LT Code II under Iterative Decoding]{
    \includegraphics[width=0.35\textwidth]{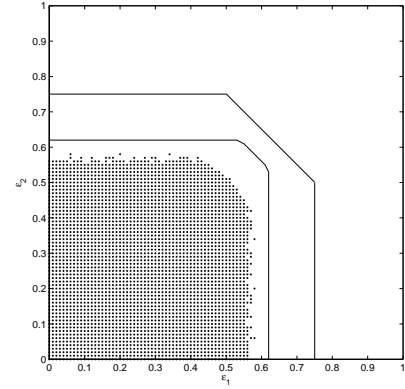}
    \label{fig:lt2_it}
  }
  \subfigure[ACPR for LT Code II under ML Decoding]{
    \includegraphics[width=0.35\textwidth]{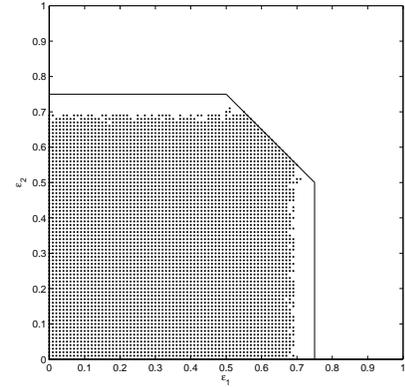}
    \label{fig:lt2_ml}
  }
  \caption{Performance of LT Code II}
\end{figure}

\subsection{LDPC Codes}
\label{sec:ldpc-codes}
Here, we consider the performance of a punctured $(4,6)$ LDPC code for the joint source-channel coding problem. 
Two systematic codes were chosen from the ensemble
LDPC$(4,6)$, and the systematic bits are punctured before transmission,
resulting in a code of rate $1/2$ (two encoded bits are transmitted per
source bit).\par
The $(4,6)$ codes achieve the entire capacity region under ML decoding,
as shown in Fig. \ref{fig:36_ml}, and the iterative decoding threshold
is significantly lower as seen in Fig. \ref{fig:36_it}. Again, this shows that
the iterative decoder is the main reason for the loss of universality.
\begin{figure}[htb]
  \centering
  \subfigure[ACPR for a $(4,6)$ LDPC Code under Iterative Decoding]{
    \includegraphics[width=0.35\textwidth]{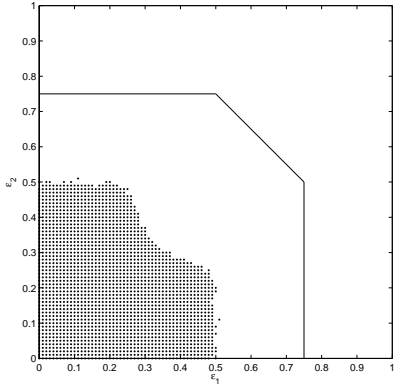}
    \label{fig:36_ml}
  }
  \subfigure[ACPR for a $(4,6)$ LDPC Code under ML Decoding]{
    \includegraphics[width=0.35\textwidth]{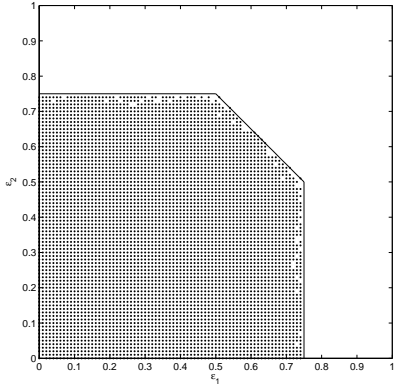}
    \label{fig:36_it}
  }
  \caption{Performance of $(4,6)$ LDPC Codes}
\end{figure}


  \section{Conclusions and Future Work}
  \label{sec:results-future}

  In this paper, we considered the performance of graph based codes with iterative
  decoding for obtaining universal performance when transmitting correlated sources
  over binary erasure channels. We designed an LT code which can achieve the extremal symmetric
  point. We then showed that an LT code optimized for the symmetric
  sum-rate point cannot achieve a corner point of the capacity
  region and, hence, we concluded that LT codes cannot be universal for this two user
  Slepian-Wolf problem. Our simulation results indicate that a punctured LDPC code ensemble
  and LT ensemble are nearly universal with maximum likelihood decoding.\par
  For future work, we plan to do the following.
  \begin{itemize}
  \item Analyze the performance of a carefully
  designed protograph code to try and achieve universality with iterative decoding.
\item Since ML decoding is nearly universal  and iterative
  decoding is not universal, we would like to see if there is an enhancement to iterative decoding that can be nearly
  universal but is yet significantly less complex than ML decoding.
  \end{itemize}

  \appendices
    \section{Proof of Theorem \ref{thm-sum-rate}}
  \label{sec:proof-theorem-refthm}
  We will use the following Lemma to show that the density evolution equations converge to zero at
  the extremal symmetric point. \newline
  \begin{lem}
    \label{sum-rate}
    \begin{equation*}
      \rho^N(x)> \frac{\mu+\rho(x)}{\mu+G_N(p)+1}, \text{ for } 0\leq x<1-\frac{1}{N}.
    \end{equation*}
   \end{lem} 
   \begin{proof}
  For $0\leq x<1-\frac{1}{N}$, we have
  \begin{align*}
      \rho^N(x) &= \frac{\mu+\sum_{i=1}^{N} \frac{\sum_{k=0}^{i-1}\binom{2i-1}{k}p^k}{i(1+p)^{2i-1}}x^i+x^N}{\mu+G_N(p)+1}\\
      &= \frac{\mu+\rho(x)+x^N}{\mu+G_N(p)+1}-\\
      &\phantom{=-}\frac{\sum_{i=N+1}^{\infty} \frac{\sum_{k=0}^{i-1}\binom{2i-1}{k}p^k}{i(1+p)^{2i-1}}x^i}{\mu+G_N(p)+1}\\
      &> \frac{\mu+\rho(x)}{\mu+G_N(p)+1}\\
 \intertext{ The last step follows from the fact that}
 \sum_{i=N+1}^{\infty} \frac{\sum_{k=0}^{i-1}\binom{2i-1}{k}p^k}{i(1+p)^{2i-1}}x^i &< \sum_{i=N+1}^{\infty}\frac{x^i}{i} \\
      &< \frac{1}{N+1}\sum_{i=N+1}^{\infty}x^i\\
      &= \frac{1}{N+1}\cdot\frac{x^{N+1}}{1-x}\\
      &< x^N,
    \end{align*}
  where the last step follows from explicit calculations, taking into
  account that $0\leq x<1-\frac{1}{N}$.      
   \end{proof}
   From \eqref{eq:one-de}, the convergence criteria for the density
   evolution equation is given by
  \begin{equation*}
    \begin{split}
      x >  \left[(1-p)+p\bar{\lambda}^N(\epsilon,x)\right]\bar{\lambda}^N(\epsilon,x)
    \end{split}
  \end{equation*}
  Consider the term $\bar{\lambda}^N(\epsilon,x) =
  \lambda\left(1-(1-\epsilon)\rho^N(1-x)\right)$. We have,
  \begin{equation*}
    \begin{split}
      \bar{\lambda}^N(\epsilon,x) &= e^{-\alpha(1-\epsilon)\cdot\rho^N(1-x)}     \\
      &\leq e^{-\alpha(1-\epsilon)\frac{\mu+\rho(1-x)}{\mu+G_N(p)+1}},\textnormal{ if } x\geq\frac{1}{N}\\
      &< e^{-\mu}\cdot \frac{\sqrt{(1-p)^2+4px}-(1-p)}{2p}\\
      &< \frac{\sqrt{(1-p)^2+4px}-(1-p)}{2p},
    \end{split}
  \end{equation*}
  where the first inequality follows from Lemma \ref{sum-rate}. The
  polynomial $f(y)=py^2+(1-p)y-x$ is a convex function of $y$, with the
  only positive root at $y = \frac{\sqrt{(1-p)^2+4px}-(1-p)}{2p}$.  So,
  if $y < \frac{\sqrt{(1-p)^2+4px}-(1-p)}{2p}$, then $f(y)<0$. Hence,
  $\left[(1-p)+p\bar{\lambda}(\epsilon,x)\right]\bar{\lambda}(\epsilon,x)-x<0$
  and the density evolution equation converges, as long as
  $x\geq\frac{1}{N}$. So, the probability of erasure is upper bounded by
  $1/N$.\par
  The rate of the code is computed as 
  \begin{equation*}
    R = \frac{\int_0^1\lambda(x)\,\text{d}x}{\int_0^1\rho^{(N)}(x)\,\text{d}x}.
  \end{equation*}
  We have
  \begin{align*}
    \int_0^1\rho^N(x)\,\text{d}x &= \frac{\mu+\sum_{i=1}^{N}
    \frac{\sum_{k=0}^{i-1}\binom{2i-1}{k}p^k}{i(i+1)(1+p)^{2i-1}}x^i + \frac{1}{N+1}}{\mu+G_N(p)+1} \\
    \intertext{also}
    \lim_{N\to \infty}\int_0^1\rho^N(x)\,\text{d}x &= \int_0^1\rho(x)\,\text{d}x \\
    &= 1-\frac{p}{2}
  \end{align*}
  Note that $\int_0^1\rho^{(N)}(x)\,\text{d}x$ is a monotonically
  increasing sequence, upper bounded by $1-\frac{p}{2}$. So, in the
  limit of infinite blocklengths the design rate is given by 
  \begin{equation*}
    R = \frac{(1-\epsilon)(1-e^{-\alpha})}{\mu+ (1-\frac{p}{2})}.
  \end{equation*}
 
\section{Proof of Theorem \ref{non-universal}}
\label{sec:proof-theorem-refnon}
To analyze the convergence of the ensemble
LDGM$\left(n,\lambda(x),\rho^{(N)}(x)\right)$, consider the functions
  \begin{equation*}
    \begin{split}
      f(x,y)=\left[(1-p)+p\bar{\lambda}^N(\epsilon_2,y)\right]\bar{\lambda}^N(\epsilon_1,x) - x\\
      g(x,y)=\left[(1-p)+p\bar{\lambda}^N(\epsilon_1,x)\right]\bar{\lambda}^N(\epsilon_2,y) - y.
    \end{split}
  \end{equation*}
  The condition for convergence of the density evolution equations are
  given by $f(x,y)<0$ and $g(x,y)<0$. When $\epsilon_1<\epsilon_2$, we
  can approximately characterize the convergence by analyzing the
  condition $g(0,y)<0$. We have
  \begin{align*}
      g(0,y) &= \left[(1-p)+p\lambda(\epsilon_1)\right]\lambda\left(1-(1-\epsilon_2)\rho^N(1-y)\right) - y\\
      & < \left[(1-p)+p\lambda(\epsilon_1)\right]\lambda\left(1-(1-\epsilon_2)\rho(1-y)\right) - y\\
      & = k\left(\sqrt{1+ay}-1\right)^\beta-y,\\
      \intertext{where}
      k &= \left(\frac{\textnormal{e}^{-\mu}(1-p)}{2p}\right)^{\frac{1-\epsilon_2}{1-\epsilon_0}}\left[(1-p)+p\textnormal{e}^{-\alpha(1-\epsilon_1)}\right],\\
      \beta &= \frac{1-\epsilon_2}{1-\epsilon_0} \text{ and } a = \frac{4p}{(1-p)^2}
  \end{align*}
  The fixed point of $g(0,y)$ can be found by solving
  \begin{align*}
    y &= k\left(\sqrt{1+ay}-1\right)^\beta, \text{ i.e.,}\\ 
    \sqrt{1+ay} &= 1+k^{-1/\beta}y^{1/\beta} \\
    \intertext{This equation is of the form}
    k^{-2/\beta}y^{(2/\beta-1)} &+2k^{-1/\beta}y^{(1/\beta-1)}-a=0,\\
    \intertext{the root of which is approximately equal to the root of the quadratic}
    k^{-2/\beta}z^2 &+2k^{-1/\beta}z-a,
  \end{align*}
  where $z = y^{(1/\beta-1/2)}$.  The positive root of the quadratic is
  given by $z = \frac{-1+\sqrt{1+a}}{k^{-1/\beta}}$. So, the fixed point
  of density evolution is $y\approx
  \left(\frac{2p}{(1-p)k^{-1/\beta}}\right)^\frac{2\beta}{2-\beta} =
  \left(\frac{2p}{(1-p)k^{-1/\beta}}\right)^{2(1-p)} =
  \left(\textnormal{e}^{-\mu}\left[(1-p)+p\textnormal{e}^{-\alpha(1-\epsilon_1)}\right]\right)^{2(1-p)}>0$.\par
  Due to the presence of a constant fixed point, which does not approach
  $0$ even in the limit of infinite maximum degree, the residual erasure
  rate is always bounded away from $0$. So, the ensemble
  LDGM$\left(n,\lambda(x),\rho^{(N)}(x)\right)$ cannot converge at a
  corner point of the capacity region.


  \balance
  \bibliographystyle{IEEEtran}

\end{document}